\documentclass[a4paper, USenglish, cleveref, autoref, thm-restate]{lipics-v2021}

\hideLIPIcs
\pdfoutput=1 %
\nolinenumbers %

\usepackage{tikz-cd}
\usepackage{listings}
\usepackage{upquote}
\usepackage{booktabs}
\usepackage{adjustbox}
\usepackage{newfloat}

\usepackage{pdfcomment}
\usepackage{marginnote}
\let\marginpar\marginnote

\newcommand{\veryshortarrow}[1][3pt]{\mathrel{%
   \hbox{\rule[\dimexpr\fontdimen22\textfont2-.2pt\relax]{#1}{.4pt}}%
   \mkern-4mu\hbox{\usefont{U}{lasy}{m}{n}\symbol{41}}}}

\usepackage{color}
\definecolor{languagecolor}    {rgb}{0.5, 0.5, 0.5} %
\definecolor{leankeywordcolor}  {rgb}{0.7, 0.1, 0.1} %
\definecolor{scalakeywordcolor}{rgb}{0.1, 0.1, 0.6} %
\definecolor{samplekeywordcolor}{rgb}{0.6, 0.6, 0.1} %
\definecolor{tacticcolor}      {rgb}{0.0, 0.1, 0.6} %
\definecolor{commentcolor}     {rgb}{0.4, 0.4, 0.4} %
\definecolor{sortcolor}        {rgb}{0.1, 0.5, 0.1} %
\definecolor{attributecolor}   {rgb}{0.7, 0.1, 0.1} %

\definecolor{symbolcolor}      {rgb}{0.0, 0.1, 0.6} %
\definecolor{nosymbolcolor}    {rgb}{0.0, 0.0, 0.0} %
\newcommand{\symbolcol}{\color{symbolcolor}}

\newcommand{\bdiamond}[1][fill=black]{\tikz [x=1ex,y=1ex,line width=.1ex,line join=round, yshift=-0.285ex] \draw  [#1]  (0,.5) -- (.5,1) -- (1,.5) -- (.5,0) -- (0,.5) -- cycle;}%

\lstdefinelanguage{lean}  {keywordstyle=[1]{\color{leankeywordcolor}},
  morekeywords=[1]{
  match, with, end, if, then, else, let, in, class, instance,
  forall, exists, Pi, fun,
  do, where, return, for, sum,
  def, inductive, CoInductive, Fixpoint, Equations, Class, Reserved, Notation, Infix, Instance,
  },
  morecomment=[l][\itshape \color{commentcolor}]{--},
  basicstyle={\ttfamily\small}}
\lstdefinelanguage{scala}{keywordstyle=[1]{\color{scalakeywordcolor}},
  morekeywords=[1]{
  import, export, protected, private, public, override, infix, extension,
  val, var, def, class, object, package, trait,
  given, using, with, extends, deriving, implicit, summon,
  match, case,
  if, then, else, break, continue, return, try, catch, for, yield, do,
  macro,
  purify,
  }}
\lstdefinelanguage{sample}{keywordstyle=[1]{\color{scalakeywordcolor}},
  morekeywords=[1]{
  do,od,map,val,mk,pure,bind,extr,link,link2,link3
  }}
\newcommand{\languagetag}{}

\makeatletter
\lst@AddToHook{DeInit}{\languagetag}
\makeatother

\lstnewenvironment{leanlisting}[1][]
  { \renewcommand{\languagetag}{\parbox[c][0pt]{\linewidth}{\raggedleft\scriptsize\color{languagecolor} Lean}}%
    \lstset{language=lean, basicstyle={\ttfamily\small}, #1}}{}
\lstnewenvironment{pseudolisting}[1][]
  { \renewcommand{\languagetag}{\parbox[c][0pt]{\linewidth}{\raggedleft\scriptsize\color{languagecolor}}}%
    \lstset{language=lean, basicstyle={\ttfamily\small}, #1}}{}
\lstnewenvironment{numpylisting}[1][]
  { \renewcommand{\languagetag}{\parbox[c][0pt]{\linewidth}{\raggedleft\scriptsize\color{languagecolor} NumPy}}%
    \lstset{language=python, basicstyle={\ttfamily\footnotesize}, #1}}{}
\lstnewenvironment{langlisting}[1][]
  { \renewcommand{\languagetag}{\parbox[c][0pt]{\linewidth}{\raggedleft\scriptsize\color{languagecolor} \Lang}}%
    \lstset{language=lean, basicstyle={\ttfamily\footnotesize}, #1}}{}
\lstnewenvironment{ainflisting}[1][]
  { \renewcommand{\languagetag}{\parbox[c][0pt]{\linewidth}{\raggedleft\scriptsize\color{languagecolor} \AINF}}%
    \lstset{language=lean, columns=[l]flexible, basicstyle={\ttfamily\footnotesize}, #1}}{}
\lstnewenvironment{myequations}[1][]
  {\renewcommand{\symbolcol}{\color{nosymbolcolor}}\noindent\adjustbox{center,vspace=\bigskipamount}\bgroup\lstset{#1, frame=none}}
  {\egroup\renewcommand{\symbolcol}{\color{symbolcolor}}}

\lstMakeShortInline[columns={[l]fullflexible},language=lean]|

\DeclareCaptionType{eqnlisting}[Equations][List of Equations]
\crefname{eqnlisting}{Equations}{Equations}
\Crefname{eqnlisting}{Equations}{Equations}

\AtBeginDocument{
  \DeclareCaptionSubType{lstlisting}
  \captionsetup[sublstlisting]{margin=3pt}
  \crefalias{sublstlisting}{listing}
}

\usepackage{bbm}

\usepackage{subcaption}
\usepackage{mathpartir}
\usepackage{stmaryrd}
\usepackage{colonequals}
\usepackage{ dsfont }
\usepackage{mathtools}
\usepackage{tabularx}
\usepackage{graphicx}

\newcommand{\Lang}{\textsc{Polara}} %
\newcommand{\AINF}{A\textsubscript{\textrm{\textit{i}}}NF}

\newcommand{\Fin}[1]{\mathsf{fin}~#1}

\title{Compiling with Arrays} %

\author%
  {David Richter}%
  {Technical University of Darmstadt, Germany}%
  {david.richter@tu-darmstadt.de}
  {https://orcid.org/0000-0002-8672-0265}%
  {}
\author%
  {Timon Böhler}%
  {Technical University of Darmstadt, Germany}%
  {timon.boehler@tu-darmstadt.de}%
  {https://orcid.org/0009-0002-9964-7367}%
  {LOEWE/4a//519/05/00.002(0013)/95}
\author%
  {Pascal Weisenburger}%
  {University of St.\texorpdfstring{\,}{ }Gallen, Switzerland}%
  {pascal.weisenburger@unisg.ch}{https://orcid.org/0000-0003-1288-1485}%
  {Swiss National Science Foundation (SNSF, No.~200429)}
\author%
  {Mira Mezini}%
  {Technical University of Damstadt, Germany\texorpdfstring{\\}{; }The Hessian Center for Artificial Intelligence (hessian.AI), Germany}%
  {mezini@informatik.tu-darmstadt.de}%
  {https://orcid.org/0000-0001-6563-7537}%
  {LOEWE/4a//519/05/00.002(0013)/95;
   HMWK cluster project \emph{The Third Wave of Artificial Intelligence} (3AI).}

\authorrunning{Richter, Böhler, Weisenburger, Mezini} %

\Copyright{Jane Open Access and Joan R. Public} %


\category{} %

\relatedversion{} %

\EventEditors{John Q. Open and Joan R. Access}
\EventNoEds{2}
\EventLongTitle{42nd Conference on Very Important Topics (CVIT 2016)}
\EventShortTitle{CVIT 2016}
\EventAcronym{CVIT}
\EventYear{2016}
\EventDate{December 24--27, 2016}
\EventLocation{Little Whinging, United Kingdom}
\EventLogo{}
\SeriesVolume{42}
\ArticleNo{23}

\begin{document}

\maketitle

\begin{abstract}
Linear algebra computations are foundational
for neural networks and machine learning, often handled through arrays.
While many functional programming languages feature lists and recursion,
arrays in linear algebra demand constant-time access and bulk operations.
To bridge this gap, some languages represent arrays
as (eager) functions instead of lists. %
In this paper, we connect this idea to a formal logical foundation
by interpreting functions as the usual negative types from polarized type theory,
and arrays as the corresponding dual positive version of the function type.
Positive types are defined to have a single elimination form
whose computational interpretation is pattern matching.
Just like (positive) product types bind two variables during pattern matching,
(positive) array types bind variables with \textit{multiplicity} during
pattern matching.
We follow a similar approach for Booleans by
introducing conditionally-defined variables.

The positive formulation for the array type enables us to
combine typed partial evaluation and common subexpression elimination
into an elegant algorithm whose result enjoys a property we call maximal fission,
which we argue can be beneficial for further optimizations.
For this purpose,
we present the novel intermediate representation \emph{indexed administrative normal form (\AINF{})},
which relies on the formal logical foundation of the positive formulation for the array type to facilitate maximal loop fission and subsequent optimizations.
\AINF{} is normal with regard to commuting conversion for both let-bindings and for-loops, leading to flat and maximally fissioned terms.
We mechanize the translation and normalization
from a simple surface language to \AINF{},
establishing that the process terminates,
preserves types, and produces maximally fissioned terms.

\end{abstract}

\section{Introduction}

Linear algebra computations are of rising importance due to their foundational role
in neural networks and other machine learning systems.
The fundamental unit of computation in linear algebra
is the multidimensional array (or just \emph{array} from now on).
Linear algebra programs are full of computations
that construct arrays from other arrays,
such as element-wise sum, matrix multiplication, convolution,
(transformer) attention, and more.

In functional programming we usually work with lists not arrays.
Lists are inductively defined data types, and processed using recursion.
Even element access on lists is implemented by recursion,
so it has to traverse the list until the element is found,
giving it running time\footnote{
  We distinguish run-time (as in run-time library)
  from running time as the time it takes to run something.}
  linear in the size of the list.
Arrays, on the other hand, should have constant-time access and
be processed using bulk operations.
As such arrays do not fit into the usual pattern of inductive data types.
A number of functional programming languages
that aim to facilitate programming of array computations have been
proposed~\cite{Gibbons17, Shaikhha19, Ragan-Kelley18}.
Aiming for high expressivity with few language constructs,
they leverage the idea that arrays can be represented as
functions~\cite{Gibbons17, Shaikhha19, Ragan-Kelley18},
or that arrays are eager functions~\cite{Paszke21}.

Functions are lazy in the sense that a 
function definition does not perform any computation and
the function body is only executed when a function is applied.
Arrays are eager in the sense that 
all contents of an array are evaluated during construction,
and array access does not perform further computation.
Prior work had this intuition on the duality between arrays and functions,
and here we ground that correspondence on proof-theoretic concepts,
which explain why arrays are eager and functions are lazy
and yield further convenient consequences.

A key insight of our work is that arrays can be interpreted
as a positively polarized version of the function type.
The connection between the positive and negative formulation of data types
and the computational interpretation of elimination forms as pattern matching
has been developed %
in the context of polarized type theory and focused
logic~\cite{Zeilberger08,Andreoli92,DownenA20, DownenA21, DownenA14, Krishnaswami09, Zeilberger09}.
Thus,
similarly to how pattern matching on a binary tuple introduces two variables,
the computational content of pattern matching on an array of size $N$ is the introduction of $N$-many variables.
Further, we explore pattern matching on Booleans by introducing \textit{conditional variables}.
This insight provides the foundation for the design of 
the \textit{indexed administrative normal form} (\AINF{}),
an intermediate representation for array computations.
We also present a simple surface array language, called
\Lang{}, and show how these changes together
(the positive formulation for the array type and a negative presentation of Booleans)
enables us to combine \emph{typed partial evaluation} (a.k.a \emph{normalization by evaluation})
and \emph{common subexpression eliminiation} into an elegant optimization algorithm
overcoming some challenges usually associated to partial evaluation.

In particular, partial evaluation of let-binding is not \emph{safe} in the sense of
always resulting in a better, or at least equal, performance than the original.
This is because a variable may appear multiple times,
hence substituting it multiple times would duplicate code.
Ideally, common subexpression elimination (CSE) would
remove redundancies introduced by partial evaluation.
But the presence of scopes, as e.g., introduced by functions, loops, and branches --
all constructs that prevail in array computations -- can complicate CSE.
While compilers can rectify these issues by using additional rules
often summarized under the general term of \textit{code motion}, 
this comes at the cost of having to decide
in what order and how often to apply these additional rules, i.e., it implies 
creating an optimization schedule, which complicates the algorithm.

Instead of these complications, with \AINF{},
we propose a novel intermediate representation for array programs
based on logical foundations that avoids the complexity of optimization schedules.
Like ANF, which is normal with regard to commuting conversion
of let-bindings implying maximal flatness,
\AINF{} is normal with regard to commuting conversions
of for-loops, thus enabling what we call maximal loop fission.
Maximal loop fission is a fundamental property
to enable further optimizations 
such as dead code elimination or common subexpression elimination.
We provide a translation of \Lang{} into \AINF{},
which performs maximal loop fission and loop invariant code motion.

\subparagraph{Contributions.}
In summary, this paper makes the following contributions:
\begin{itemize}
\item We present \AINF{}, an intermediate representation that makes use of the unconventional
idea of treating arrays as positive types from polarized type theory.
\AINF{} is normal with regard to commuting conversion for both let-bindings and for-loops,
leading to flat and maximally fissioned terms.
\item We present \Lang{}, a simple surface array language,
along with a translation of \Lang{} to \AINF{},
for which we prove termination, type preservation, and maximal fission.
\item We present an optimization algorithm for \AINF{} based on normalization by evaluation and common subexpression elimination,
for which we prove termination and type preservation.
\end{itemize}

The artifact containing our mechanization can be found anonymously at:

\url{https://figshare.com/s/cddb1c0eeaf37e42f318}

\begin{figure}
\begin{subfigure}{.25\textwidth}
\begin{pseudolisting}[basicstyle={\ttfamily\scriptsize}]
  let z =
    let y1 = x + 1
    2 * y1


  let y2 = x + 1
  ...
\end{pseudolisting}
\caption{Nested lets.}\label{fig:samples-nest}
\end{subfigure}%
\begin{subfigure}{.25\textwidth}
\begin{pseudolisting}[basicstyle={\ttfamily\scriptsize}]

  let y1 = x + 1
  let z = 2 * y1


  let y2 = x + 1
  ...
\end{pseudolisting}
\caption{Flat let-bindings.}\label{fig:samples-flat}
\end{subfigure}%
\begin{subfigure}{.25\textwidth}
\begin{pseudolisting}[basicstyle={\ttfamily\scriptsize}]
  let f = fun i:nat.
    let y1 = x + 1
    2 * y1


  let y2 = x + 1
  ...
\end{pseudolisting}
\caption{Functions.}\label{fig:samples-fun}
\end{subfigure}%
\begin{subfigure}{.25\textwidth}
\begin{pseudolisting}[basicstyle={\ttfamily\scriptsize}]
  let z = if c then
    let y1 = x + 1
    2 * y1
  else
    2 * x
  let y2 = x + 1
  ...
\end{pseudolisting}
\caption{Branches.}\label{fig:samples-if}
\end{subfigure}%

\caption{Sample programs.}
\label{fig:samples}
\end{figure}

\section{Problem Statement}

Typed partial evaluation is a powerful optimization technique~\cite{Jones93},
which can reduce excessive terms by applying computation laws. 
For example, it can reduce a projection on a pair $(a,b).1 \equiv a$ by applying $\beta$-reduction.
Or, it can eliminate superfluous branches like in
$\mathsf{if}~x~\mathsf{then}~(\mathsf{if}~x~\mathsf{then}~a~\mathsf{else}~b)~\mathsf{else}~c \;\equiv\; \mathsf{if}~x~\mathsf{then}~a~\mathsf{else}~c$
by applying uniqueness laws (e.g., $\eta$-expansion).

But, while shining on its logical foundation, partial evaluation is not a \emph{safe} optimization.
A \emph{safe} optimization has to either reduce the running time of a program
or at least preserve it.
Partial evaluation of let-bindings is not \emph{safe}
because a variable may appear multiple times,
hence substituting it multiple times would duplicate code.
Ideally, common subexpression elimination (CSE) would
remove all redundancies introduced by partial evaluation.
But scopes introduced by nested let-bindings, functions, and branches 
complicate CSE.
For illustration, below we consider a few examples of redundancies
that can occur in programs and how CSE handles them.

In \cref{fig:samples-nest}, a program with nested let-bindings is shown.
Here, the variable |z| is bound to |2 * y1|,
where |y1| is bound to the successor of |x|;
and then the variable |y2| is bound to the successor of |x| as well.
It is easy to see that |y1| is redundant with |y2|,
yet |y1| is not in scope at the definition of |y2|,
so we cannot simply replace one by the other.
The problem can be avoided by bringing the program into a form,
where no let-binding is nested inside another let-binding,
such that all previously bound variables are in scope for the whole remaining expression.
Consider the program shown in \cref{fig:samples-flat},
which is equivalent to the previous program,
but this time no expression has a subexpression.
Now, the former definition is in scope at the latter definition,
and thus |y2| can be replaced by |y1|,
thereby eliminating a duplicate subexpression.

As mentioned, functions and branches introduce scope as well,
and therefore complicate CSE.
Yet, the solution of flattening the code is not as straightforward to apply.
To illustrate the problem with functions,
consider the program shown in \cref{fig:samples-fun},
which defines a function |f|.
Inside the function the successor of |x| is bound to |y1|,
and outside the function it is bound redundantly to |y2|.
To share the expressions,
we could consider moving the definition of |y1| out of the functions.
But moving an expression out of a function is not \emph{safe},
as long as we do not know whether the function will be called at all.

To illustrate the problem with branches,
consider the program shown in \cref{fig:samples-if}.
Here, the result of a conditional expression is bound to the variable |z|,
in one branch the successor of |x| is bound to |y1|,
and after the conditional expression the successor of |x| is bound to the variable |y2|.
Similar to the function case, to share the expressions,
we could consider moving the definition of |y1| out of the branch,
but that is again not a \emph{safe} optimization,
as long as we do not know that this branch is taken.

To rectify the issues outlined above compilers use additional rules
often summarized under the general term of \textit{code motion}.
But this comes at the cost of introducing the problem of having to decide
in what order and how often to apply these additional rules
(i.e., creating an optimization schedule).

Our work avoids the complexity of optimization schedules. 
We argue that -- instead of complicating the CSE algorithm with
optimization schedules -- a better approach is to design an intermediate representation
for array programs, which like ANF bans nested expressions. 
This simplifies the optimization of array programs, which now can safely rely on algorithms 
based on logical foundations such as partial evaluation and CSE. 
The novel intermediate representation, 
called \AINF{}, is informally presented in the following along with a simple surface arrays language 
and the optimized translation of the latter to the former.

\section{\AINF{}, \Lang{}, and Simplified Optimizations}
\label{sec:challenge-solution}
We describe the two key insights 
on which our approach is based (\cref{sec:funs-and-arrays,sec:sums}),
introduce \Lang{} and \AINF{} by example (\cref{sec:byExample}),
and explain how \AINF{} simplifies optimizations (\cref{sec:simplifications}).

\subsection{The Duality of Functions and Arrays}
\label{sec:funs-and-arrays}

The list type is an inductive datatype defined by its constructors |nil| for the empty list
and |cons| for constructing a list from another list with an additional element.
Accordingly, algorithms over lists work by recursion,
expressed with functions and branches.
As element access on lists is implemented by recursion,
it has to traverse the list until the element is found,
giving it running time linear in the size of the list.
Arrays, on the other hand, enjoy constant-time access
and feature bulk operations.
The consequence is that arrays do not to fit into the usual
pattern of inductive data types.
Nevertheless,
because custom semantics would require further proofs to ensure soundness,
arrays are occasionally modelled as lists,
with the hint that the actual running time can differ
(in Lean for example\footnote{https://lean-lang.org/lean4/doc/array.html}).

In functional array languages,
we exploit the equivalence of an array of type $X$ and length $n$
with a function from a natural number below $n$ to a value of type $X$.
Forward, this equivalence allow us to access (|get|) elements of an array by its index.
Backward, we create (|tabulate|) an array
from a function describing each individual element based on their index.
The forward direction is indeed already very much ingrained in everyday programming,
as array access |a[i]| and function application |a(i)| look very much alike in many languages,
and even share identical syntax in some.

\begin{center}
$\text{Array}_n~X ~\leftrightarrow~ (\text{Fin}_n \to X)$ \\[1em]
$\text{get}:      \text{Array}_n~X \to (\text{Fin}_n \to X)$ \\
$\text{tabulate}: (\text{Fin}_n \to X) \to \text{Array}_n~X$
\end{center}

But something important changes in the conversion from a function to an array,
and vice versa. Functions are \emph{lazy},
in the sense that the evaluation of a function is delayed until it is applied, while arrays are \emph{eager},
in the sense that all elements of an array have already been evaluated
and on array access only need to be looked up. Also, in a language with (side) effects,
the two types can be distinguished in that a function application can trigger effects,
while an array access cannot trigger effects.
Dually, constructing a function cannot trigger effects,
while constructing an array can trigger effects.

We can put the relationship between functions and arrays on a logical foundation
by considering the difference between positive and negative types \cite{Zeilberger08}.
A positive type is defined by a set of constructors (introduction forms),
and we get a single corresponding destructor (elimination form) for it
with one continuation for the content of each possible constructor (pattern matching).
A negative type is defined by a set of destructors,
and we get a single corresponding constructor for it 
that has to provide one value for each destructor to extract (copattern matching).
Positive types are usually associated with eager (call-by-value) evaluation,
and negative types with lazy (call-by-name) evaluation.
Many types can be defined either as a positive or as a negative type.
For illustration, we consider the positive and the negative formulations of the product type below.

\subparagraph{Products as Positive and as Negative Types.}
The product type as a positive type $\times$
has a single constructor $(a,~b)$ (\textsc{Intro}).
A corresponding destructor (\textsc{Elim})
can be systematically derived as pattern matching on the constructor.
Reduction (\textsc{Beta}) occurs when a destructor is applied
to a constructor, and they eliminate each other.
\begin{mathpar}
\inferrule[Intro]
  {\Gamma \vdash a: A \\
   \Gamma \vdash b: B}
  {\Gamma \vdash (a,~b): A \times B}

\inferrule[Elim]
  {\Gamma \vdash p: A \times B \\
   \Gamma,~a:A,~b:B \vdash c: C}{
   \Gamma \vdash \mathsf{let}~ (a,~b) = p;~ c: C}

\inferrule[Beta]
  {\Gamma \vdash a: A \\
   \Gamma \vdash b: B \\
   \Gamma,~ x:A,~ y:B \vdash c:C}
  {\Gamma \vdash (\mathsf{let}~ (x,~y) = (a,~b);~ c) \equiv c[x:=a,~ y:=b]}
\end{mathpar}

Alternatively, products can also be defined as negatives types $\otimes$.
In this case, we give primacy to a set of destructors, namely
the projections $p.\mathsf{fst}$ and $p.\mathsf{snd}$ (\textsc{Elim1}, \textsc{Elim})
to access the individual elements of a tuple $p$,
and derive systematically the corresponding constructor (\textsc{Intro})
providing one value for each destructor to extract.
Beta reduction occurs (\textsc{Beta1}, \textsc{Beta1})
when a destructor is applied to a constructor
by extracting the corresponding value.

\begin{mathpar}
\inferrule[Elim1]{\Gamma \vdash p: A \otimes B}{\Gamma \vdash p.\mathsf{fst}: A}

\inferrule[Elim2]{\Gamma \vdash p: A \otimes B}{\Gamma \vdash p.\mathsf{snd}: B}

\inferrule[Intro]{\Gamma \vdash a: A \\ \Gamma \vdash b: B}{\Gamma \vdash (\mathsf{fst} = a;~ \mathsf{snd} = b) : A \otimes B}
\\
\inferrule[Beta1]{\Gamma \vdash a: A \\ \Gamma \vdash b: B}{\Gamma \vdash (\mathsf{fst} = a, \mathsf{snd} = b).\mathsf{fst} = a}

\inferrule[Beta2]{\Gamma \vdash a: A \\ \Gamma \vdash b: B}{\Gamma \vdash (\mathsf{fst} = a, \mathsf{snd} = b).\mathsf{snd} = b}
\end{mathpar}

\subparagraph{Functions as Negative and Positive Types.}
Usually, the function type is considered a negative type.
It has a single destructor -- function application $f\,a$ (\textsc{Elim}) --
and the corresponding constructor is systematically derived
by copattern matching on the possible destructors (\textsc{Intro}).
When a destructor is applied to the constructor,
we extract the value provided as the body of the function,
and substitute the variable with the argument (\textsc{Beta}).

\begin{mathpar}
\inferrule[Elim]{\Gamma \vdash f: A \to B \\ \Gamma \vdash a: A}{\Gamma \vdash f~a: B}

\inferrule[Intro]{\Gamma, a: A \vdash b: B}{\Gamma \vdash \mathsf{fun}~ a.~ b : A \to B}

\inferrule[Beta]{\Gamma, x: A \vdash b: B \\ \Gamma \vdash a: A}{\Gamma \vdash (\mathsf{fun}~x.~b)~a \equiv b[x:=a]}
\end{mathpar}

The function type can also be represented as a positive type.
In this case, the function is primarily defined through its constructor,
and the destructor is systematically derived from pattern matching on the constructor.
But the interpretation of positive function types comes with some challenges for the metatheory.
The introduction form of a function turns a term-in-the-context-of-a-variable $a:A \vdash b:B$
into a function $(\mathsf{fun}~a.~b): A \to B$.
Thus, the corresponding elimination form of a function $(\mathsf{fun}~a.~b): A \to B$
should introduce a variable of type term-in-the-context-of-a-variable $a:A \vdash b:B$ into the context.
But to properly model that, we need a judgment where we have a context in the context,
in other words a ``higher-order judgment''~\cite{nlab:function_type, Nanevski08}.
A judgment is higher-order
when an entailment $\vdash$ occurs inside the context of another entailment
An implementation of higher-order judgments needs to ensure
that a variable which has such a judgment as a type
is only used in larger contexts,
where all required variables are available.
For example, $b: (a: A \vdash B) \vdash b: B$ is invalid,
given that $b$ must occur in a context where an $a:A$ is available;
while $b: (a: A \vdash B) \vdash (\mathsf{fun}~a{:}A.~ b): A \to B$ is valid,
because a variable $a:A$ has been introduced such that the use of $b$ afterwards is safe.
\begin{mathpar}
\inferrule[Intro]
  {\Gamma, a: A \vdash b: B}
  {\Gamma \vdash (\mathsf{fun}~a.~b): A \to B}

\inferrule[Elim]
  {\Gamma \vdash f: A \to B \\ \Gamma, x:(a: A \vdash B) \vdash c: C}
  {\Gamma \vdash \mathsf{let}~ (\mathsf{fun}~ a.~ x) = f;~ c : C}

\inferrule[Beta]
  {\Gamma, a:A \vdash b: B \\ \Gamma, x:(a: A \vdash B) \vdash c: C}
  {\Gamma \vdash \mathsf{let}~(\mathsf{fun}~a.~x) = (\mathsf{fun}~a.~b);~ c \equiv c[x:=b]}
\end{mathpar}

\subparagraph{Interpreting positive function types as arrays.}

We avoid the challenges of interpreting functions as positive types by proposing 
to interpret positive function types as arrays,
re-interpreting the rules of the positive function type as the rules of the array type.
We require the argument type to be the type of natural numbers below some number |n|,
which corresponds to the index type of an array.
Traditionally, functional programming works with lists and not arrays,
therefore the constant-time access of arrays is not accurately represented by the model;
in functional array languages~\cite{Henriksen17, Shaikhha19}, arrays tend to live in the shadow of the function type,
as their introduction and elimination forms depend on (higher-order) functions.
Interpreting the array as a positive function type makes them independent
and puts them on an equal footing to the other types with regard to their logical foundation.

We distinguish positive functions, i.e., arrays, from normal functions
by using $\Rightarrow$ for the type,
writing $(\mathsf{for}~ a.~ b)$ as the introduction form for arrays,
while the elimination form is given, as always,
by pattern matching on all possible introduction forms:
\begin{mathpar}
\inferrule*[lab=Intro, right=\textup{$A = \text{Fin}_n$}]
  {\Gamma, x: A \vdash b: B}
  {\Gamma \vdash (\mathsf{for}~x.~b): A \Rightarrow B}

\inferrule*[lab=Elim, right=\textup{$A = \text{Fin}_n$}]
  {\Gamma \vdash f: A \Rightarrow B \\ \Gamma, x:(a: A \vdash B) \vdash c: C}
  {\Gamma \vdash \mathsf{let}~ (\mathsf{for}~ a.~ x) = f;~ c : C}

\inferrule[Beta]
  {\Gamma, a:A \vdash b: B \\ \Gamma, x:(a: A \vdash B) \vdash c: C}
  {\Gamma \vdash \mathsf{let}~(\mathsf{for}~a.~x) = (\mathsf{for}~a.~b);~ c \equiv c[x:=b]}
\end{mathpar}

Intuitively,
analogously to how pattern matching on a product
introduces two variables (one for each projection of the product),
pattern matching on an array introduces a family of variables,
one for each element.
For illustration, consider
  $b[a:=2]$
as
  $b_2$,
and
  $(\mathsf{let}~ (\mathsf{for}~ a.~ b) = f;~ c)$
as
  $(\mathsf{let}~ (b_0, b_1, ..., b_{n-1}) = f;~ c)$.

\subparagraph{Arrays Enable CSE.}
Flattening let-bindings usually helps CSE.
More precisely the rule that is used to create the ANF representation is the let-let commuting conversion:
\begin{pseudolisting}
(let y = (let x = e1; e2); e3) ≡
(let x = e1; let y = e2; e3)
\end{pseudolisting}

Note that the following let-fun commuting conversion is not \emph{safe}
because on the left-hand side |e1| is evaluated at most once,
even if it was used multiple times in |e2|;
but on the right-hand side it will be evaluated once for each usage in |e2|.
\begin{pseudolisting}
(let y = (fun i. let x = e1; e2); e3) ≡
(let (fun i. x) = (fun i. e1); let y = (fun i. e2); e3)
\end{pseudolisting}

On the other hand, the let-for commuting conversion that we use in \AINF{} below is \emph{safe}
and states that the following two lines are equivalent.
As array construction is evaluated eagerly,
the expression |e1| is evaluated just once for each iteration,
on both sides of the equation.
\begin{pseudolisting}
(let y = (for i. let x = e1; e2); e3) ≡
(let (for i. x) = (for i. e1); let y = (for i. e2); e3)
\end{pseudolisting}

Intuitively, this rule allows us to split a complex loop into multiple simpler loops,
hence it is closely connected to loop \emph{fission}.
If we use this rule to split every loop as much as possible, then we end up with a normal
form in which every loop only contains a single variable.
First performing loop fission as much as possible
helps with implementing other optimizations,
for example allows CSE to remove redundancies that it could not otherwise eliminate.

The use of this rule means that frequently both the left side and the right side
of a variable definition are surrounded by the same form
(on the left as a pattern form, on the right as a term form), so we will introduce some syntactic sugar and write
|let for i. (x = e1); e2| to mean |let (for i. x) = (for i. e1); e2| in the following.

\subsection{Lifting Branching into the Context}
\label{sec:sums}

A different problem arises with values of the Boolean type,
Booleans have two constructors, $\mathsf{true}$ and $\mathsf{false}$
(\textsc{Intro1}, \textsc{Intro2}).
They have one destructor, the conditional expression (\textsc{Elim}),
where one continuation is provided for each constructor,
the consequent and the alternative.
A conditional reduces to the consequent when the condition is true (\textsc{Beta1}),
and to the alternative when the condition is false (\textsc{Beta2}).
\begin{mathpar}
\inferrule[Intro1]
  {~}
  {\Gamma \vdash \mathsf{true}: \mathsf{bool}}

\inferrule[Intro2]
  {~}
  {\Gamma \vdash \mathsf{false}: \mathsf{bool}}

\inferrule[Elim]
  {\Gamma \vdash p: \mathsf{bool} \\
   \Gamma \vdash e: C \\
   \Gamma \vdash f: C}
  {\Gamma \vdash \mathsf{if}~p~\mathsf{then}~e~\mathsf{else}~f}

\inferrule[Beta1]
  {\Gamma \vdash e : C \\ \Gamma \vdash f : C }
  {\Gamma \vdash \mathsf{if}~\mathsf{true}~\mathsf{then}~e~\mathsf{else}~f \equiv e}

\inferrule[Beta2]
  {\Gamma \vdash e : C \\ \Gamma \vdash f : C }
  {\Gamma \vdash \mathsf{if}~\mathsf{false}~\mathsf{then}~e~\mathsf{else}~f \equiv f}
\end{mathpar}

The let-if commuting conversion below is \emph{safe} with regard to running time.
But applying the commuting conversion duplicates the expression |e3|.
If |e3| is a big expression,
then even if duplicating it in different branches may not impact the running time,
having nested branches will lead to a blow-up of the code size exponential in the number of branches
(which is also bad for the compiling time).

\begin{pseudolisting}
(let z = (if e0 then e1 else e2); e3) ≡
(if e0 then (let z=e1; e3) else (let z=e2; e3))
\end{pseudolisting}

Essentially, the above rule bubbles up conditionals to the top of the expression.
Instead, we propose to trickle down the conditionals
using conditionally defined variables and state a new |let-if| commuting conversion 
that does not duplicate the branches.
To avoid duplicating branches, we introduce a syntactically single-branch |if|.
A single-branch |if| produces a conditional value, i.e., a value that can only be accessed if the condition is true, and the single-branch |if!| produces a value that can only be accessed if the condition is false.

Analogously, we have let-bindings for conditional variables |let (if e. x) := ...| (or |let (if! e. x) := ...|), which define a variable |x| that can only be accessed if the condition |e| is true (or false, respectively).
Two simple syntactical conditions can be used to check whether a conditional variable is accessible:
First, a conditional variable is accessible on the right-hand side of the definition of another conditional variable that has the same condition.
In other words, a conditional variable can be used to define the value of another conditional variable with the same condition. %
Second, a conditional variable is accessible in one of the branches of a standard two-branched |if| condition.
In other words, two mutually exclusive conditional variables can be combined with an |if| to define a non-conditional variable.

Using the single-branch if, we can now express a |let-if| commuting conversion, that does duplicate |e3|. Here the double-branched if is seperated into two single-branch ifs and |e3| remains to be executed once afterwards:
\begin{pseudolisting}
(let z = (if e0 then e1 else e2);
 e3)
≡
(let (if  e0. z1) = (if e0. e1);
 let (if! e0. z2) = (if! e0. e2);
 let          z   = (if e0 then z1 else z2);
 e3)
\end{pseudolisting}

Correct use of conditional values will thus frequently lead to the use of the same condition 
on the variable bound by a |let| and in a single-branched if in the bound expression. 
We will thus make use of syntactic sugar writing |let if e0. (x1 = e1); e2| to mean |let (if e0. x1) = (if e0. e1); e2|.

\subsection{\Lang{} and \AINF{} by Example}
\label{sec:byExample}

In this section, we informally introduce both \Lang{} and \AINF{}. 
We do so by giving examples of array operations and programs in \Lang{}
and showing the result of compiling them to \AINF{}. 
Just for reference, we will also provide versions of the \Lang{} examples 
written in the widely-adopted array programming library NumPy~\cite{Harris20}. 
Please note that the focus of this paper lies in the exploration of \AINF{}, rather than \Lang{}. 
The latter serves merely as a vehicle to elucidate how \AINF{} effectively facilitates 
optimizations during the translation process from a surface array language. 
Hence, compared to NumPy, we have purposely kept it closer to low-level imperative code.

In \Lang{}, an expression |e| is either a constant |c|, or an arithmetic operator |e + e|, function application |e e|, array access |a[i]|, array construction |for i:n. e|, or summation |sum i:n. e|, as well as pairs |(e, e)| and projection |e.1|, |e.2|.
The array construction |for i:n. e| constructs an array of length |n| by repeatedly evaluating |e|, with |i| bound to the values from |0| to |n-1|.
For example, |for i:3. 10*i| evaluates to |[0, 10, 20]|.
We will write |for i. e| if the size of the array can be inferred from the context.
Summation is syntactic sugar for constructing an array and then summing it, so |sum i.e ≡ sum (for i. e)|.

\begin{table}
\caption{Common linear algebra operations in \Lang{}; NumPy for reference. }
\label{tbl:linalgops}
\begin{tabular}{l l l l} \toprule
  Name & NumPy & \Lang{} \\ \midrule
  Vector addition & |v + w| & |for i. v[i] + w[i]| \\ 
  Matrix addition & |A + B| & |for i j. A[i,j] + B[i,j]| \\ 
  Element-wise product (vector) & |v * w| & |for i. v[i] * w[i]| \\ 
  Element-wise product (matrix) & |A * B| & |for i j. A[i,j] * B[i,j]| \\
  Outer product & |np.multiply.outer(A, B)| & |for i j k l. A[i,j] * B[k,l]| \\
  Trace & |A.trace()| & |sum i. A[i, i]| \\ 
  Transpose & |A.transpose()| & |for i j. A[j, i]|  \\ 
  Matrix multiplication & |A @ B| & |for i k. sum j. A[i,j] * A[j,k]| \\
  Matrix-vector multiplication & |A @ v| & |for i. sum j. A[i,j] * v[j]| \\ \bottomrule
\end{tabular}
\end{table}

In \cref{tbl:linalgops}, we list several common linear algebra operations and
compare how they can be expressed using the linear algebra library NumPy and \Lang{}.
We assume as given that the vectors |v, w| and matrices
|A, B| are of appropriate sizes.

\subparagraph{Dense Layer.}

As a slightly more involved example,
we show how a dense neural network layer can be implemented
in NumPy, \Lang{}, and \AINF{}, respectively.
The NumPy example makes use of the built-in matrix multiplication
operator |@|. While such an operation can be implemented as function in \Lang{}, 
we show an example that only relies on the few \Lang{} primitives,
using the |for| looping construct and indexing.
Likewise, while the NumPy definition uses the built-in |maximum| function and addition, 
the \Lang{} version uses an explicit loop
that performs element-wise multiplication and additions across the vectors.

Compared to the untyped NumPy program, 
we also declare the types of the arguments. A type \lstinline{n=>flt} describes
an array of floating point numbers with size \lstinline{n}.

Obviously, the corresponding \AINF{} program (Figure \ref{fig:ainfdense}) is rather lengthy,
as every intermediate result gets assigned to a variable,
just like in ANF.

\noindent
\begin{minipage}[t]{.42\textwidth}
\begin{numpylisting}
def dense(b, W, x):
  return np.maximum(0, W @ x + b)
\end{numpylisting}
\end{minipage}
\begin{minipage}[t]{.58\textwidth}
\begin{langlisting}
dense(b:n=>flt, W:n=>m=>flt, x:m=>flt): n=>flt :=
  for i. max(0, (sum j. W[i][j] * x[j]) + b[i])
\end{langlisting}
\end{minipage}

\subparagraph{Convolution.}

We now describe how to express convolution in \Lang{}.
Convolution involves moving a vector, called the kernel, across another vector
while repeatedly calculating the dot product.
For this example, we need to subtract two indices to indicate that we shift one array
while keeping the other as it is.
In the \AINF{} example (Figure \ref{fig:ainfconv}), we create a two-dimensional array |x10| containing all the
possibilities for shifting the array |y|. For example if |y = [1,2,3]|, then |x10| is a matrix of size |3×3|
so that |tmp1 = [[1,2,3], [2,3,1], [3,1,2]]|. We then form the dot product of each entry with |x|.

\noindent
\begin{minipage}[t]{.5\textwidth}
\begin{numpylisting}
def conv(x, y):
  return np.convolve(x, y, 'same')
\end{numpylisting}
\end{minipage}
\begin{minipage}[t]{.5\textwidth}
\begin{langlisting}
conv(x: n=>flt, y: (n+m-1)=>flt): m=>flt :=
  for i. sum j. x[j] * y[j+i]
\end{langlisting}
\end{minipage}

\subparagraph{Black-Scholes.}
Black-Scholes is a simplified mathematical model for the dynamics of derivative investments in financial markets.
The Black-Scholes formula provides an estimate for the price of the call option (buying) and the put option (selling) of a European-style option given the original price $S$, the strike price $K$, the expiration time $T$, the force-of-risk $r$ and the standard deviation $\sigma$.
The interesting part, from an array programming language's perspective, is that with a naive implementation of the calls and puts as separate functions, common subexpression elimination is not able to identify the redundant computation across these functions over two separate loops.

In particular, note the redundant definition of |d1| and |d2| in the |calls| and the |puts| function. This code gets inlined into the |blackScholes| function, but the two function calls land in separate loops. Nevertheless, using loop fission the output can be reduced to just 22 lines of \AINF{}
(see Figure \ref{fig:ainfbs}); without fission and CSE the generated code would have 54 lines.

\begin{langlisting}
calls(S: flt, K: flt, T: flt: r: flt: sigma: flt): flt :=
  let d1 := (log (S / K) + (r + sigma * sigma / 2) * T) / (sigma * sqrt T)
  let d2 := d1 - sigma * sqrt T
  S * normCdf d1 - K * exp (0 - var r * var T) * normCdf d2

puts(S: flt, K: flt, T: flt: r: flt: sigma: flt): flt :=
  let d1 := (log (S / K) + (r + sigma * sigma / 2) * T) / (sigma * sqrt T)
  let d2 := d1 - sigma * sqrt T
  K * exp (0 - r * T) * normCdf (0 - d2) - S * normCdf (0 - d1)

blackScholes(arr: (n => flt)): n => (flt × flt) :=
  let S := 1; let K := 1; let r := 1; let sigma := 1
  let Calls: (n => flt) := for i. calls(S, K, arr[i], r, sigma)
  let Puts:  (n => flt) := for i. puts(S,  K, arr[i], r, sigma)
  for i. (Calls[i], Puts[i])
\end{langlisting}

\begin{figure}
\begin{minipage}[t]{.5\textwidth}
\begin{subfigure}[t]{\linewidth}
\begin{ainflisting}[basicstyle={\ttfamily\scriptsize}]
dense(b: n => flt, W: n => m => flt,
      x: m => flt): n => flt :=
let for i:n,      (x0 : flt      := 0          )
let for i:n, j:m, (x1 : m => flt := W[i]       )
let for i:n, j:m, (x2 : flt      := x1[j]      )
let for i:n, j:m, (x3 : flt      := x[j]       )
let for i:n, j:m, (x4 : flt      := x2 * x3    )
let for i:n,      (x5 : m => flt := for j:m, x4)
let for i:n,      (x6 : flt      := sum x5     )
let for i:n,      (x7 : flt      := b[i]       )
let for i:n,      (x8 : flt      := x6 + x7    )
let for i:n,      (x9 : flt      := max x0 x8  )
let               (x10: n => flt := for i:n, x9)
x10
\end{ainflisting}
\caption{\AINF{} for a dense layer.}
\label{fig:ainfdense}
\end{subfigure}
\begin{subfigure}[t]{\linewidth}
\begin{ainflisting}[basicstyle={\ttfamily\scriptsize}]
conv(x: n => flt, y: p => flt): m => flt :=
let for i:m, j:n, (x0 : flt      := x[j]       )
let for i:m, j:n, (x1 : fin p    := j + i      )
let for i:m, j:n, (x2 : flt      := y[x1]      )
let for i:m, j:n, (x3 : flt      := x0 * x2    )
let for i:m,      (x4 : n => flt := for j:n, x3)
let for i:m,      (x5 : flt      := sum x4     )
let               (x6 : m => flt := for i:m. x5)
x6

  where p = n+m-1
\end{ainflisting}
\caption{\AINF{} for convolution.}
\label{fig:ainfconv}
\end{subfigure}
\end{minipage}
\begin{minipage}[t]{.5\textwidth}
\begin{subfigure}[t]{\linewidth}
\begin{ainflisting}[basicstyle={\ttfamily\scriptsize}]
blackScholes(arr: n => flt):  n => flt × flt :=
let for i1:n, (x0  : flt := 1.500000          )
let for i1:n, (x1  : flt := i0[i1]            )
let for i1:n, (x2  : flt := x0 * x1           )
let for i1:n, (x4  : flt := sqrt x1           )
let for i1:n, (x5  : flt := x2 / x4           )
let for i1:n, (x6  : flt := normCdf x5        )
let for i1:n, (x7  : flt := 0.000000          )
let for i1:n, (x9  : flt := x7  - x1          )
let for i1:n, (x10 : flt := exp x9            )
let for i1:n, (x19 : flt := x5  - x4          )
let for i1:n, (x20 : flt := normCdf x19       )
let for i1:n, (x21 : flt := x10 * x20         )
let for i1:n, (x22 : flt := x6  - x21         )
let for i1:n, (x37 : flt := x7  - x19         )
let for i1:n, (x38 : flt := normCdf x37       )
let for i1:n, (x39 : flt := x10 * x38         )
let for i1:n, (x47 : flt := x7  - x5          )
let for i1:n, (x48 : flt := normCdf x47       )
let for i1:n, (x49 : flt := x39 - x48         )
let for i1:n, (x50 : (flt × flt) := (x22, x49))
let (x51 : (n => flt × flt) := for i1:1, x50  )
x51
\end{ainflisting}
\caption{\AINF{} for a Black-Scholes.}
\label{fig:ainfbs}
\end{subfigure}
\end{minipage}
\caption{Generated \AINF{}.}
\label{fig:ainfgen}
\end{figure}

\subsection{Simplifying Optimizations with \AINF{}}
\label{sec:simplifications}

This section provides an overview showing how some classical optimizations can be applied to \AINF{}
and illustrates why our novel normal form simplifies their implementation.

To improve readability, we will sometimes present \AINF{} code in a way that deviates from the actual representation by putting multiple operations in one line,
when this does not affect the optimization.

\subparagraph{Loop Fission.}
Loop fission is an optimization pass that prepares code to improve the effectiveness of dead code elimination.
Standard dead code elimination can only delete whole loops.
Loop fission splits a loop into parts, so that individual parts that are not used 
can be removed.
In \AINF{}, the body of each loop is a single variable,
which means that any \AINF{} program is necessarily as fissioned as possible.
A simple partial evaluation pass on \AINF{} can then perform dead-code elimination.

Below, the left side shows an array computation in \Lang. On the right, that same computation has been transformed to \AINF{},
which implies loop fission.
The code follows the principle from ANF that expressions should be atomic, i.e., 
only have one operation. In terms of array programming, this leads to the first loop being split into three loops.
Partial evaluation could then reduce the full program to just |for i. f(xs[i])|.

\noindent\begin{minipage}{.45\linewidth}
\begin{langlisting}
let x = for i.
  let ys = f(xs[i])
  let zs = f(xs[i])
  (ys, zs)
for i.
  fst x[i]
\end{langlisting}
\end{minipage}
\hspace{.1\linewidth}
\begin{minipage}{.45\linewidth}
\begin{ainflisting}
let for i. (ys = f(xs[i]))
let for i. (zs = f(xs[i]))
let for i. (x  = (ys, zs))
let for i. (y  = fst x)
let z = for i. y
z
\end{ainflisting}
\end{minipage}

A further advantage of loop fission is that it improves loop \textit{fusion}:
Splitting a program into as many loops as possible, gives more freedom to the algorithm
for combining loops again.

\subparagraph{Common subexpression elimination.}

We can now see how \AINF{} helps with CSE.
On the left, we recapitulate the example from above; on the right, we can see the
same program in \AINF{}.

\noindent\hspace{.11\linewidth}\begin{minipage}{.33\linewidth}
\begin{langlisting}
let f = for i.
  let y = x+1
  let y' = 2*y
  y'
let z = x+1
...
\end{langlisting}
\end{minipage}
\hspace{.11\linewidth}
\begin{minipage}{.33\linewidth}
\begin{ainflisting}
let for i. (y  = x + 1)
let for i. (y' = 2 * y)
let for i. (f  = y')

let for i. (z  = x + 1)
...
\end{ainflisting}
\end{minipage}

The loop computing |f| has been broken down into two loops. As a result, |z| is clearly redundant, as it performs the same computation
in the same scope as |y|. Therefore, compared to array languages using higher-order functions,
\AINF{} allows us to use the simple, standard approach to CSE,
and nonetheless remove redundancies between expressions inside and outside of loops.

\subparagraph{Loop invariant code motion.}
Loop invariant code motion (LICM), which moves constants out of a loop, is another optimization that benefits from \AINF{}.
In \AINF{}, this would correspond to dropping an unused index;
hence, the implementation of LICM is very simple.
On the left, we generate an array |ys|, in which every element is the constant |1|.
We then compute an array |zs| that makes use of |ys|.
Notice that the index |i| that is bound in the creation of |ys| is not used.
We can therefore eliminate that loop, 
adjusting uses of |ys| accordingly from |ys[i:=j]| to |ys|, as seen on the right.
\\

\noindent\hspace{.11\linewidth}\begin{minipage}{.33\linewidth}
\begin{ainflisting}
let for i. (ys = 1)
let zs = for i. f(xs[i], ys)
...
\end{ainflisting}
\end{minipage}
\hspace{.11\linewidth}
\begin{minipage}{.33\linewidth}
\begin{ainflisting}
let ys = 1
let zs = for i. f(xs[i], ys)
...
\end{ainflisting}
\end{minipage}

\section{Mechanization}

\begin{figure}
\begin{center}
$n \in \mathbb{N} ~~~~ f \in \mathbb{F} ~~~~ x \in \text{Var} ~~~~ i \in \text{Idx}$ \\[1em]
$\begin{array}{lrcl}
\text{Types} & t &\Coloneqq&
  \Fin{n} \mid \mathsf{flt} \mid t \;\hat\times\; t \mid t \;\hat\to\; t \mid n \;\hat\Rightarrow\; t \\
\text{Constants} & c &\Coloneqq&
  n \mid f \mid \hat+ \mid \hat\cdot \mid \hat- \mid \hat/ \mid \mathsf{app} \mid \mathsf{get} \mid \mathsf{pair} \mid \mathsf{fst} \mid \mathsf{snd} \mid \mathsf{sum} \\
\text{\Lang{}} & e &\Coloneqq&
  c~\overline{e} \mid x \mid \mathsf{fun}~x{:}t.~e \mid \mathsf{for}~x{:}n.~ e \mid \mathsf{ite}~e~e~e \mid \mathsf{let}~ e;~ e \\[1em]
\text{\AINF{}} & a &\Coloneqq&
  \mathsf{let}~ C[x = p];~ a \mid x \\
\text{Primitives} & p &\Coloneqq&
  c~\overline{x} \mid i \mid \mathsf{fun}~i{:}t.~x \mid \mathsf{for}~i{:}n.~ x \mid \mathsf{ite}~x~x~x \\
\text{Scope Contexts} & C[\cdot] &\Coloneqq&
  \cdot \mid C[\mathsf{fun}~i{:}t.~\cdot] \mid C[\mathsf{for}~i{:}n.~\cdot] \mid C[\mathsf{if}~x{\not=}0.~\cdot] \mid C[\mathsf{if}~x{=}0.~\cdot]
\end{array}$
\end{center}
\caption{\Lang{} and \AINF{}}
\label{fig:syntax}
\end{figure}

\begin{figure}
\begin{mathpar}
    \inferrule{n < m}{\vdash n : \Fin{m}}\quad
    \vdash f : \mathsf{flt} \quad
    \vdash \mathsf{app} : (t_1 \hat\to t_2) \to t_1 \to t_2 \quad
    \vdash \mathsf{get} : (n \hat\Rightarrow t_1) \to \Fin{n} \to t_1 \\
    \vdash \mathsf{pair} : t_1 \to t_2 \to (t_1 \hat\times t_2) \quad
    \vdash \mathsf{fst} : (t_1 \hat\times t_2) \to t_1 \quad
    \vdash \mathsf{snd} : (t_1 \hat\times t_2) \to t_2 \\
    \vdash \hat + : \Fin{n} \to \Fin{m} \to \Fin{(n+m-1)} \quad
    \vdash \hat + : \mathsf{flt} \to \mathsf{flt} \to \mathsf{flt} \quad
    \vdash \mathsf{sum} : (n \hat\Rightarrow \mathsf {flt}) \to \mathsf{flt}  \\
    \inferrule[Var]{x{:}t \in \Gamma}{\Gamma \vdash x : t} \quad
    \inferrule[Const]{\vdash c : \overline{t_i \to}~t' \quad
    \overline{\Gamma \vdash e_i : t_i}}{\Gamma \vdash c~\overline{e_i} : t'}\\
    \inferrule[Fun]{\Gamma, x{:}t_1 \vdash e : t_2}{\Gamma \vdash \mathsf{fun}~x{:}t_1.~e : t_1 \to t_2} \quad
    \inferrule[For]{\Gamma, i{:}\Fin{n} \vdash e_2 : t}
    {\Gamma \vdash \mathsf{for}~i{:}n.~e_2 : n \hat\Rightarrow t} \\
    \inferrule[Let]{\Gamma \vdash e_1 : t_1 \quad \Gamma, x{:}t_1 \vdash e_2 : t_2}
    {\Gamma \vdash \mathsf{let}~x = e_1;~ e_2 : t_2} \quad
    \inferrule[Ite]{\Gamma \vdash e_1 : \Fin{2} \quad \Gamma \vdash e_2 : t \quad \Gamma \vdash e_3 : t}
    {\Gamma \vdash \mathsf{ite}~e_1~e_2~e_3 : t}
\end{mathpar}
\caption{\Lang{}'s type system.}
\label{fig:types}
\end{figure}

We mechanized \Lang{}, partial evaluation of \Lang{},
\AINF{}, the translation from \Lang{} to \AINF{},
and common subexpression elimination over \AINF{},
using the dependently typed programming language Lean 4~\cite{Moura21}.

\subsection{\Lang{} and Partial Evaluation}
\subparagraph{\Lang{}.}
The \Lang{} grammar uses the set of natural numbers, floating point numbers, variables and indices (Figure~\ref{fig:syntax}), but the distinction between variables and indices is only relevant for \AINF{}.
Types are floating point numbers, products, functions, and arrays,
as well as bounded natural numbers, i.e. $\mathsf{fin}~n$ is the type
of numbers smaller than $n$.
Constants are natural number and floating point literals,
arithmetic symbols,
function application ($\mathsf{app}$),
array access ($\mathsf{get}$),
pair construction ($\mathsf{pair}$),
first and second projection ($\mathsf{fst}$, $\mathsf{snd}$),
and array summation ($\mathsf{sum}$).
\Lang{} terms are variable access, n-ary constant application,
function abstraction, array construction, branching ($\mathsf{ite}$), and let-binding.
We decided to put first-order syntax forms such as function application, array access, pairing,
and the product projections into the constants,
because they are all handled uniformly by the following algorithms,
while the higher-order syntax forms, i.e., the ones that bind variables,
such as function abstraction, array construction,
branching, and let-binding are kept in the terms
because they are all treated differently.

\subparagraph{Intrinsic Types.}
The typing rules for \Lang{} are given in Figure \ref{fig:types}.
First, we give the types for constants.
Note that we use the $\to$ symbol for the typing judgement of constants that take arguments.
This is not to be confused with the type constructor $\hat\to$.
The typing rules for variables (\textsc{Var}),
function abstractions (\textsc{Fun}),
and |let|-bindings (\textsc{Let}) are standard.
The rule \textsc{Const} allows one to apply a constant to a number (possibly zero) of arguments.
For example, as \textsf{app} has type $(t_1 \hat\to t_2) \to t_1 \to t_2$,
the expression $\mathsf{app}~e_1~e_2$ has type $t_2$ when $e_1 : t_1 \hat\to t_2$ and
$e_2 : t_1$. The \textsc{For} rule shows that constructing an array with \textsf{for} requires
an expression of type \textsf{nat} for the size and another expression, which can use the
(numerical) index \textsf{i} and whose type gives the element type of the array.
The \textsc{Ite} rule states that the condition has to be of type $\mathsf{nat}$ and the two branches have to be of the same type (the condition is considered true if nonzero).

\subparagraph{PHOAS.}
Our formal development uses parametric higher-order abstract syntax (PHOAS)~\cite{PfenningE88,Chlipala08},
allowing us to leverage the binders of the host language as binders for the guest language.
Terms are parametrized by an abstract denotation of types $\Gamma$,
and variables contain a value of that type.
By using PHOAS, we can avoid certain technicalities relating to variable binding such as
capture-avoiding substitution, thereby streamlining the implementation.

\subparagraph{Static Size.}
As mentioned above, our array types have the form \lstinline{n => a}, where \lstinline{n} is the size of the array.
The fact that the size of an array is always part of its static type, implies that the
sizes of all arrays are known at compile time.
This guarantees that indexing can be statically checked for out-of-bounds array accesses, ensuring
the absence of run time errors without requiring run time checks.

Because \Lang{} is not polymorphic, expressions operating on arrays are fixed to specific
array sizes. For example, there is no single expression in \Lang{} that can map a function
over an array of \textit{arbitrary} size.
This restriction is alleviated because our language is embedded, allowing
us to reuse polymorphism from the host language. More concretely, we can define a function in the
host language that for each number \lstinline{n} returns a \Lang{} term implementing \lstinline{map} on an array of size \lstinline{n} (here, $\lambda$ belongs to the host language
and \lstinline{fun} belongs to \Lang{}):
\[
  \begin{array}{lll}
  \text{map} &:& (n : \mathds{N}) \to (\Gamma \vdash (t_1 \hat\to t_2)~\hat\to~(n\hat\Rightarrow t_1)~\hat\to~(n\hat\Rightarrow t_2)) \\
  \text{map} &:=& \lambda n.~\mathsf{fun~f~a.~for~i.~f~a[i]}
  \end{array}
\]

\subparagraph{Termination.}
The use of static array sizes ensures that array indexing is total.
In fact, every language construct in \Lang{} is deterministic and terminating, making the language total; hence it is not Turing-complete.
Most notably, we eschew general recursion in favor of the more well-behaved looping construct \lstinline{for}.
The lack of non-termination allows us to give a simple denotational semantics
and guarantees termination of normalization, as described next.

\subparagraph{Normalization by Evaluation (NbE).}
Normalization is defined in Figure~\ref{fig:norm} by a denotation for types ($\llbracket t \rrbracket_\Gamma : \mathrm{Type}$),
a corresponding denotation for terms and constants (such that
when $e$ has type $t$, then $(\llbracket e \rrbracket_\Gamma : \llbracket t 
\rrbracket_\Gamma)$), as well as functions |quote| ($\eta$), |splice| ($\eta'$), and |norm|.
Note the additional argument $\Gamma$ -- this is a peculiarity of the PHOAS representation, where $\Gamma$ determines the denotation of variables.
This argument can take different values, depending on which information we want to extract from a term.
For example, when pretty-printing a term we want to produce a string,
so we associate every variable also with a string $(\Gamma~t := \text{String})$.
For NbE,
every term should be translated to the denotation of their type,
so we associate every variable to the denotation $\llbracket\cdot\rrbracket_\Gamma$ of its type using $\Gamma$.
The function |norm| takes a value of type $(\forall \Gamma.~\Gamma \vdash t)$
and returns one of the same type. The quantification means that we can only use
variables that were created by the language's binding constructs, so the type represents closed terms.

The denotation of the natural number is a natural number term,
the denotation of a floating point number is a floating point number term,
the denotation of a product is a product of the denotations,
the denotation of a function is a function of the denotations,
the denotation of an array is a function
from a natural number term to a denotation of the array's content.
Later, for code generation,
we will again distinguish functions and arrays.
But for the purpose of normalization by partial evaluation (NbE),
we model arrays as functions so as to reduce the need for rules for both of them.

\begin{figure}
\begin{subfigure}[b]{.45\textwidth}
\begin{center}
\[\begin{array}{lcl}
\llbracket \cdot \rrbracket_\Gamma & : & \mathrm{Ty} \to \mathrm{Type} \\
\llbracket \Fin{n} \rrbracket_\Gamma & = & \Gamma \vdash \Fin{n}\\
\llbracket \mathsf{flt} \rrbracket_\Gamma & = & \Gamma \vdash \mathsf{flt}\\
\llbracket t_1 \;\hat\times\; t_2 \rrbracket_\Gamma & = & \llbracket t_1 \rrbracket_\Gamma \times \llbracket t_2 \rrbracket_\Gamma \\
\llbracket t_1 \;\hat\to\; t_2 \rrbracket_\Gamma & = & \llbracket t_1 \rrbracket_\Gamma \to \llbracket t_2 \rrbracket_\Gamma \\
\llbracket n \;\hat\Rightarrow\; t \rrbracket_\Gamma & = & (\Gamma \vdash \mathsf{nat}) \to \llbracket t \rrbracket_\Gamma \\
\end{array}\]
\end{center}
\caption{Denotation of types.}
\end{subfigure}
\begin{subfigure}[b]{.55\textwidth}
$~~\llbracket \mathsf{ite}~e_1~e_2~e_3 \rrbracket =$
\vspace{-0.2em}
\begin{center}
$
  \begin{cases}
  \llbracket e_2 \rrbracket & \text{if}~\llbracket e_1 \rrbracket = 1 \\
  \llbracket e_3 \rrbracket & \text{if}~\llbracket e_1 \rrbracket = 0 \\
  \eta' (\mathsf{ite}~\llbracket e_1 \rrbracket~(\eta \llbracket e_2 \rrbracket)~
                          (\eta \llbracket e_3 \rrbracket)) & \text{otherwise}
  \end{cases}
$
\end{center}
\caption{Denotation of \textsf{ite}.}
\end{subfigure}
\begin{subfigure}[b]{.45\textwidth}
\begin{center}
\[\begin{array}{ll}
\llbracket \cdot \rrbracket &: (\llbracket \cdot \rrbracket_\Gamma \vdash t) \to \llbracket t \rrbracket_\Gamma \\
\llbracket x \rrbracket &=
  x \\
\llbracket \mathsf{fun}~i.~e \rrbracket &=
  \lambda i.~\llbracket e~i \rrbracket \\
\llbracket \mathsf{for}~i.~e \rrbracket &=
  \lambda i.~\llbracket e~i \rrbracket \\
\llbracket c~\overline{e} \rrbracket &=
  \llbracket c \rrbracket ~ \overline{\llbracket e \rrbracket} \\
\llbracket \mathsf{let}~e_1;~e_2 \rrbracket &=
  \llbracket e_2 \rrbracket ~ \llbracket e_1 \rrbracket  \\[1em]
\llbracket \mathsf{app} \rrbracket~e_1~e_2 &= e_1~e_2 \\
\llbracket \mathsf{get} \rrbracket~e_1~e_2 &= e_1~e_2 \\
\llbracket \mathsf{pair} \rrbracket~e_1~e_2 &= (e_1,~e_2) \\
\llbracket \mathsf{fst} \rrbracket~e &= e.1 \\
\llbracket \mathsf{snd} \rrbracket~e &= e.2 \\
\llbracket \mathsf{sum} \rrbracket~e &= \eta'~(\mathsf{sum}~(\eta~e)) \\
\llbracket \hat+ \rrbracket~n_1~n_2 &= n_1 + n_2 \\
\llbracket \hat+ \rrbracket~e_1~e_2 &= e_1 ~\hat+~ e_2 \\
\end{array}\]
\end{center}
\caption{Denotation of terms and constants.}
\end{subfigure}
\begin{subfigure}[b]{.55\textwidth}
\begin{center}
\[
\begin{array}{ll}
\eta &: \forall t.~\llbracket t \rrbracket_\Gamma \to (\Gamma \vdash t) \\
\eta~(t_1~\hat\to~t_2)~ e &= \mathsf{fun}~i:t_1.~\eta~t_2~(e~(\eta'~t_1~i)) \\
\eta~(n_1~\hat\Rightarrow~t_2)~ e &= \mathsf{for}~i:n_1.~\eta~t_2~(e~(\eta'~t_1~i)) \\
\eta~(t_1~\hat\times~t_2)~ e &= \mathsf{tup}~(\eta~t_1~e.1)~(\eta~t_2~e.2) \\
\eta~\mathsf{nat}~ e &= e \\
\eta~\mathsf{flt}~ e &= e \\[1em]
\eta' &: \forall t.~(\Gamma \vdash t) \to \llbracket t \rrbracket_\Gamma \\
\eta'~(t_1~\hat\to~t_2)~ e &= \lambda i.~\mathsf{app}~(\eta'~t_2~e)~(\eta~t_1~i) \\
\eta'~(n_1~\hat\Rightarrow~t_2)~ e &= \lambda i.~\mathsf{get}~(\eta'~t_2~e)~(\eta~t_1~i) \\
\eta'~(t_1~\hat\times~t_2)~ e &= (\eta'~t_1~(\mathsf{fst}~e),~\eta'~t_2~(\mathsf{snd}~e)) \\
\eta'~\mathsf{nat}~ e &= e \\
\eta'~\mathsf{flt}~ e &= e \\[1em]
\text{norm} &: (\forall \Gamma.~\Gamma \vdash t) \to (\forall \Gamma.~\Gamma \vdash t) \\
\text{norm}~e &= \eta~\llbracket e \rrbracket
\end{array}
\]
\end{center}
\caption{Quote $\eta$, splice $\eta'$, and normalization $\text{norm}$. }
\end{subfigure}

\caption{Typed partial evaluation.}
\label{fig:norm}
\end{figure}

The quote $\eta$ and splice $\eta'$ functions
perform eta-expansion of terms
by recursion over the types.
Quote turns denotations into terms, and splice turns terms back into denotations.
The denotation of a \Lang{} term is a corresponding host-language value of that term (i.e., a Lean value in our mechanization).
NbE then evaluates terms
in the environment of splicing,
followed by quoting the denotation back into a term.

Constants denote functions that check for whether their argument is known,
and the partial evaluation of their argument; otherwise, they quote/splice
the term into a denotation of the type.

\subsection{\AINF{} and Common Subexpression Elimination}
\subparagraph{FOAS.}
An essential component for implementing common subexpression elimination
is the ability to compare to terms for equality.
As we cannot decide equality over functions, we have to convert
from parametric higher-order abstract syntax (PHOAS)
to first-order syntax (FOAS) to get decidable equality for identifiers and terms containing variables.

\subparagraph{\AINF{}.}
In \AINF{}, we distinguish between variables $x$ and indices $i$ (Figure~\ref{fig:syntax}).
Variables are introduced by let-binding,
while indices are introduced by functions and loops.
An \AINF{} term is a sequence of pattern-matching let-bindings of primitives,
ending in a final variable (Figure~\ref{fig:syntax}, \AINF{}).
An essential property of \AINF{} is thus,
that it is both maximally fissioned (each for loop just has a single variable as a body)
and maximally flat (an \AINF{} term is a single list of terms without subterms,
executed one after another).
Pattern matching contexts $C$ have one hole for the variable,
and one form for each higher-order argument to any term former,
namely array construction, function abstraction, if-consequence, and if-alternative.
Primitives are constant application, indices, variable access,
function abstraction, and array construction.

Conversion to \AINF{} (Figure~\ref{fig:fission}) exploits the fact that continuation-passing-style
automatically flattens code.
The function $C\llbracket ~e~ \rrbracket~k$ takes as inputs a \Lang{} term $e$,
a pattern matching context $C$, and a continuation $k$,
and returns an \AINF{} term.
In the mechanization the function uses a reader monad as well
to generate unique variable names. %
The function is initialized with the empty pattern matching context,
and the identity continuation; the variable counter is initialized with zero.

Another important helper function is smart binding $\llparenthesis ~p~ \rrparenthesis~k$,
which takes a primitive $p$ and a continuation $k$.
Smart binding ensures that every primitive term passed to it is
bound to a variable name, and that variable name is passed to the continuation.
If the primitive term is a variable already,
this variable name is passed to the continuation;
otherwise the term is bound to a unique variable name, incrementing the counter.

Translation to \AINF{} by $C\llbracket ~e~ \rrbracket~k$
recurses structurally over the term $e$.
In the case of a variable or an index, the term is forwarded to smart binding.
In the case of a constant application (exemplary shown for binary constant application),
first the first subterm is translated, then in the continuation the second subterm is translated, and in the continuation the term is reconstructed as a primitive with variable referencing the name of the translated subterms, which is passed to smart binding to generate a new name for this term, passing the continuation along.
In the case of function abstraction, array construction, and conditional expressions,
the subterms are translated as well, but in adapted contexts,
and the final term is passed as well to smart binding to generate a name for it,
and the continuation is passed along.
Concretely, in the case of function abstraction, the function body is translated in a context which includes the function argument.
In the case of array construction, the array body is translated in a context which includes the iteration variable.
In the case of conditional expressions, the consequence is translated in a context which includes the condition, and the alternative is translated in a context which includes the negation of the condition.
Finally, in the case of let-binding, first the right-hand side of the binding is translated, and then the body of the binding.

\begin{figure}
\begin{subfigure}{1\textwidth}
\[\begin{array}{ll}
C \llbracket ~\cdot~ \rrbracket~\cdot &: (\vdash t_1) \to (\text{Var}~t_1 \to~\text{\AINF{}}~t_2) \to~\text{\AINF{}}~t_2 \\
C \llbracket ~x~ \rrbracket~k &=
  \llparenthesis~ x ~\rrparenthesis~k \\
C \llbracket ~i~ \rrbracket~k &=
  \llparenthesis~ i ~\rrparenthesis~k \\
C \llbracket ~c~e_1~e_2~ \rrbracket~k &=
  C \llbracket e_1 \rrbracket~ \lambda~x_1.~~%
  C \llbracket e_2 \rrbracket~ \lambda~x_2.~~%
  \llparenthesis~ c~x_1~x_2 ~\rrparenthesis~k \\
C \llbracket \mathsf{fun}~i{:}t.~e \rrbracket~k &=
  C[\mathsf{fun}~i{:}t.~\cdot] \llbracket e \rrbracket~\lambda~x.~~%
  \llparenthesis~ \mathsf{fun}~i{:}t.~x ~\rrparenthesis~k \\
C \llbracket \mathsf{for}~i{:}e_1.~e_2 \rrbracket~k &= 
  C \llbracket e_1 \rrbracket~\lambda~x_1.~~%
  C[\mathsf{for}~i{:}x_1.~\cdot] \llbracket e_2 \rrbracket~\lambda~x_2.~~%
  \llparenthesis~ \mathsf{for}~i{:}x_1.~x_2 ~\rrparenthesis~k \\
C \llbracket \mathsf{ite}~e_1~e_2~e_3 \rrbracket~k &=
  C \llbracket e_1 \rrbracket~ \lambda~x_1.~~%
  C[\mathsf{if}~x_1{=}0.~\cdot]     \llbracket e_2 \rrbracket~\lambda~x_2.~~%
  C[\mathsf{if}~x_2{\not=}0.~\cdot] \llbracket e_3 \rrbracket~\lambda~x_3.~~%
  \llparenthesis \mathsf{ite}~x_1~x_2~x_3 \rrparenthesis~k \\
C \llbracket \mathsf{let}~e_1;~e_2 \rrbracket~k &=
  C \llbracket e_1 \rrbracket~\lambda~x_1.~~%
  C \llbracket e_2~x_1 \rrbracket~k \\
\end{array}\]
\caption{Fission.}
\end{subfigure}
\begin{subfigure}{1\textwidth}
\[\begin{array}{lll}
\llparenthesis~\cdot~\rrparenthesis &:& \text{Prim}~t_1 \to (\text{Var}~t_1 \to \text{\AINF{}}~t_2) \to \text{\AINF{}}~t_2 \\
\llparenthesis~ x ~\rrparenthesis~k &=& k~x \\
\llparenthesis~ p ~\rrparenthesis~k &=& \mathsf{let}~x=p;~k~x ~~~~~ \text{where} ~x~ \text{unique}
\end{array}\]
\caption{Smart binding.}
\end{subfigure}
\caption{Fission with smart binding.}
\label{fig:fission}
\end{figure}

\subparagraph{CSE.}
In addition to deciding equality for terms,
a further complication with common subexpression elimination is that
we also need to decide equality in the presence of already established equalities.
For example consider the term |x=v, y=v, z=(x+y), q=(y+x), t=z+1, r=q+1, |$\ldots$.
Correct CSE should eliminate it to |x=v, z=(x+x); t=z+1; rename [y->x; q->z; r->t], |$\ldots$.
Notice how the later eliminations are dependent on the earlier ones.
If we simply rename the remaining term every time
we detect a variable to be redundant,
then this algorithm would perform exponentially worse,
because every renaming is a traversal over the whole remaining term, and
CSE itself is already a traversal over the whole term.
To keep everything with a single traversal,
we adapt CSE to carry a renaming with it,
which is applied just before a term is checked for redundancy.

CSE (Figure~\ref{fig:cse}) takes a renaming, a naming, and an \AINF{} term,
and returns a new \AINF{} term of the same type.
A renaming is a list of pairs of variables of the same type,
representing that the first variable is to be replaced by the second.
A naming is a list of pairs of a primitive and a variable of the same type,
representing that the primitive term has been previously bound to that variable.
CSE works by structural recursion over the term.
When the input term is just a variable, it simply applies the renaming.
When the input term is a let-binding, then the renaming is applied to the term as well.
The renamed term is looked up in the list of previously defined terms.
If the term has not been bound to a variable name already ($\text{none}$),
then the term is now let-bound to a variable, inside a renamed pattern matching context C.
CSE proceeds with the remaining terms $a$,
remembering that the term $p'$ has been bound to $\sigma$,
so that future redundant occurrences of $p'$ can be eliminated.
If the term has already been bound to a variable name  ($\text{some}~x'$),
then no let-binding is produced,
but only the renaming is extended to replace future references to $x$
to the already existing $x'$ instead.
CSE proceeds with the remaining terms $a$,
remembering that the term $p'$ has been bound to $\sigma$,
so that future redundant occurrences of $p'$ can be eliminated.

\begin{figure}
\[
\begin{array}{lll}
  \text{Ren} &=& [(t_1:\text{Ty})\times \text{Var}~t_1\times \text{Var}~t_1] \\
  \text{Nam} &=& [(t_1:\text{Ty})\times \text{Prim}~t_1\times \text{Var}~t_1] \\
  \text{CSE} &:& \text{Ren} \to \text{Nam} \to \text{\AINF{}}~t_2 \to \text{\AINF{}}~t_2 \\
  \text{CSE}~r~\sigma~x &=& (\sigma,~\text{ren}_r~x) \\
  \text{CSE}~r~\sigma~(\mathsf{let}~C[x=p];~a) &=&
    \begin{cases}
    \text{CSE}~r~(\mathsf{let}~C'[x=p']; \sigma)~a \\
      ~~~~~~~~~~~~~~~~~~~~~ \text{if}~~\text{lookup}~\sigma~C'~x~p' = \text{none} \\[1em]
    \text{CSE}~([x:=x'] :: r)~\sigma'~a \\
      ~~~~~~~~~~~~~~~~~~~~~ \text{if}~~ \text{lookup}~\sigma~C'~x~~p' = \text{some}~(x', \sigma')
    \end{cases} \\
  && \text{where}~ p' = \text{ren}_r~p  \\
  && \text{where}~ C' = \text{ren}_r~C  \\
\end{array}
\]
\caption{Common subexpression elimination.}
\label{fig:cse}
\end{figure}

\subsection{Mechanization in Lean}

In this section, we present excerpts from the Lean mechanization and relate them to the paper formalization. 
The type of terms |Tm| corresponds to $(\Gamma \vdash t)$ and features constructors for variables and constants (|var|, |cst0| etc.).
In the paper, we do not write these constructors explicitly, so we would write |x| rather than |var x|.
We define the following types corresponding to the above definitions of syntax (Figure~\ref{fig:syntax}) in Lean.

\begin{leanlisting}
inductive Var : Ty → Type -- Variables Var
inductive Par : Ty → Type -- Indices   Idx

inductive Ty                           -- Types t
inductive Const0 : Ty → Type           -- Constants c (nullary)
inductive Const1 : Ty → Ty → Type      -- Constants c (unary)
inductive Const2 : Ty → Ty → Ty → Type -- Constants c (ternary)
inductive Tm (Γ: Ty → Type): Ty → Type -- Terms e

inductive Prim : Ty → Type -- Primitives p
inductive Env  : Type      -- Scoped Contexts C
inductive AINF : Ty → Type -- AINF a
\end{leanlisting}

In particular, we define the following functions in Lean.
The function |Ty.de| corresponding to denotation of types $\llbracket \cdot \rrbracket_\Gamma$,
|quote| to $\eta$ and |splice| to $\eta'$,
|Const0.de|, |Const1.de|, |Const2.de| and |Tm.de| were shown as term, constant, and its denotations $\llbracket \cdot \rrbracket$.
Finally, |norm| is defined using |term| denotations and |quote|.

\begin{leanlisting}
def Ty.de (Γ : Ty → Type): Ty → Type

def quote {Γ} : {α : Ty} → Ty.de Γ α → Tm Γ α
def splice {Γ} : {α : Ty} → Tm Γ α → Ty.de Γ α

def Const0.de : Const0 α → Ty.de Γ α
def Const1.de : Const1 β α → Ty.de Γ β → Ty.de Γ α
def Const2.de : Const2 γ β α → Ty.de Γ γ → Ty.de Γ β → Ty.de Γ α
def Tm.de : Tm (Ty.de Γ) α → Ty.de Γ α

def Tm.norm : (∀ Γ, Tm Γ α) → Tm Γ α
  | e => quote (Tm.de (e _))
\end{leanlisting}

The |smart_bnd| function takes an additional number argument, wrapped inside a reader monad,
which is used for creating fresh variables.
In the paper, we leave this out and just stipulate that the variable is fresh.
The same applies to |toAINF|.
When discussing CSE in the paper, we describe renamings.
The |rename| functions define how a renaming is applied.
CSE also requires us to check equality of expressions, which is done with the |beq| functions.
The CSE function in the paper also calls |lookup|, which is not defined there.
It corresponds to the built-in |ListMap.lookup|.
Our code also contains a function |Env.or|, which merges two environments. This
is used to allow CSE to remove redundancies which appear in different, but compatible,
environments.

In particular, we define the following functions in Lean, corresponding to the functions above:

\begin{leanlisting}
def Prim.beq : Prim α → Prim α → Bool
def AINF.beq : AINF α → AINF α → Bool
def AINF.smart_bnd : Env → Prim α → (VPar α → Counter (AINF β)) → Counter (AINF β)
def Tm.toAINF (e : Tm VPar α) : AINF α
def Var.rename : Ren → Var α → Var α
def VPar.rename (r: Ren): VPar α → VPar α
def Env.rename (r: Ren): Env → Env
def Prim.rename (r: Ren): Prim α → Prim α
def AINF.rename (r: Ren): AINF α → AINF α
def AINF.rename (r: Ren): AINF α → AINF α
def Env.or (Γ: Env) (Δ: Env): Tern → Option Env := fun t => match Γ, Δ with
def RAINF.upgrade : RAINF → Var b → Env → Option RAINF
def AINF.cse' : Ren → RAINF → AINF α → (RAINF × VPar α)
def merge: RAINF → VPar α → AINF α

def AINF.cse : Ren → RAINF → AINF α → AINF α
  | r, σ, a => let (b, c) := a.cse' r σ; merge b.reverse c
\end{leanlisting}

\subsection{Proofs}

In this section, we show that normalization and translation to \AINF{} are type-preserving, i.e.
given a well-typed term, they always produce a valid term of the same type. We also show that
translation to \AINF{} produces maximally fissioned terms.

We use an intrinsically typed approach where the type system of the object language is included in the
encoding of the data type for the language's syntax.
Therefore, the host languages type system ensures only well-typed
terms can be constructed.

Following an intrinsically typed approach means that the soundness properties hold simply because
our (appropriately typed) definitions type check. We do not have to state and prove explicit, separate theorems,
because the types of the functions already carry the necessary information.

\begin{theorem}[Well-typedness of Optimization] \quad\newline
Our optimization procedure is terminating and type preserving.
\end{theorem}
\begin{proof}
Termination is ensured by Lean's built-in termination check.
The fact that normalization terminates relies on \Lang{} being
a total language. In particular, the absence of unbounded
recursion and the combination of static array sizes with
intrinsic typing avoids infinite loops and out-of-bounds accesses, ensuring that our normalization function
always successfully terminates.
Type preservation is ensured by intrinsically-typed mechanization; consider the types of normalization and CSE in Lean:
\begin{leanlisting}
def Tm.norm : (∀ Γ, Tm Γ α) → Tm Γ α
  | e => quote (Tm.de (e _))
def AINF.cse : Ren → RAINF → AINF α → AINF α
  | r, σ, a => let (b, c) := a.cse' r σ; merge b.reverse c
\end{leanlisting}
Intrinsic typing defines the
typing of the object language (here, \Lang{}) using the typing of the host language
(here, Lean), so the host language's type checker prevents the creation of ill-typed object language programs.
This means that an element of |(∀ Γ, Tm Γ α)| is a well-typed \Lang{}
program and an element of |AINF α| is a well-typed \AINF{} term. Further, given a well-typed term, each function returns
a well-typed term, which is what we mean by soundness with regard to the type system.
\end{proof}

\begin{theorem}[Well-typedness of Translation] \quad\newline
Our translation procedure is terminating and type preserving.
\end{theorem}
\begin{proof}
Again, termination is guaranteed by Lean's termination checker.
The argument for type preservation is similar to the one above: As both \Lang{} and \AINF{} are
defined using intrinsic typing, we can only construct well-typed programs.
Consider the type of |toAINF| (we omit the definition):
\begin{leanlisting}
def Tm.toAINF (e : Tm VPar α) : AINF α
\end{leanlisting}
If one tried to define |toAINF| in a way that produces an ill-typed program,
the definition would be rejected by the type checker. 
\end{proof}

Finally, \AINF{} is inductively defined to be maximally fissioned,
i.e., as a list of primitives without subterms,
therefore the act of translating \Lang{} terms into \AINF{}
in a total programming language performs loop fission by definition.

\begin{theorem}[Maximal Fission] \quad\newline
Our translation into \AINF{} produces terms with maximal fission.
\end{theorem}
\begin{proof}
Consider the definition of \AINF{} terms:
\begin{leanlisting}
inductive AINF : Ty → Type
  | ret : VPar α → AINF α
  | bnd : Env → Var α → Prim α → AINF β → AINF β
\end{leanlisting}
Here, a value of type |VPar α| can be a variable or a parameter. A value of
type |Prim α| is a primitive (not nested) operation.
The constructor |bnd| represents a variable assignment while |ret| returns a variable or parameter and represents the end of the program. From this inductive definition,
it is apparent that all \AINF{} terms have a flat structure where nested expressions
are impossible. Recall that, in \AINF{}, each assignment is considered its own separate loop. Because the body of each assignment only contains a single primitive, each loop has a body only consisting of one operation and hence an \AINF{} term is guaranteed to be maximally fissioned. Because the translation function |toAINF| has output type |AINF α|,
it can \textit{only} produce such maximally fissioned terms. Further, Lean's termination
checker ensures that |toAINF| is total, and so always returns an \AINF{} term
in finite time.
\end{proof}

\section{Related Work}
\label{sec:related-work}

\subsection{Intermediate Languages}

Early work by Steele~\cite{Steele78} implemented a continuation-passing-style (CPS) IR in a functional compiler,
stressing the suitability of CPS for compilation,
as it closely mimics how control flow is expressed
with jumps in hardware instructions,
and makes evaluation order explicit in the syntax.
Appel~\cite{Appel1992} %
observed that beta-reductions in the lambda calculus
are unsound in the presence of side effects as they could duplicate the effect.
Yet, CPS, which makes evaluation order explicit, enables to perform certain optimizations,
such as
dead code elimination (DCE), and
common subexpression elimination (CSE),
by exploiting that in CPS every subterm is referenced by a unique name.

Sabry and Felleisen~\cite{Sabry93}
identified that additional power of compiling in CPS~\cite{Plotkin75}
corresponds to the additional rules of the monadic computational language~\cite{Moggi89}.
Of particular importance is the so-called associativity law of the monad,
i.e., the let-let commuting conversion, enabling the flattening of code.
Then, Flanagan~\cite{FlanaganSDF93}
coined the name ``A-normal form (ANF)'' for the now popular IR,
which in contrast to CPS, expresses sequential execution
by simple let-binding rather than continuations.
The difference between ANF and the monadic language
is that ANF forbids nested let-bindings,
i.e., code must be normal with regard to the associativity rule of the monad.

However, Kennedy~\cite{Kennedy07}
showed that moving from CPS to ANF did not take into account branching.
More precisely,
while the let-let commuting conversion enables the flattening of code,
the let-if commuting conversion \emph{duplicates} code into each branch,
in the worst case leading to blow-up of code size exponential in the number of branches.
Given that recursion always includes a branch
for base case(s) and step case(s),
the same problem appears with recursion.
Kennedy therefore argued for a return to CPS.

The issue was resolved by Maurer et al.~\cite{MaurerDAJ17},
who provided an implementation of ANF for use in the Glasgow Haskell Compiler (GHC), and further simplified by Cong et al.~\cite{CongOER19}
who provided implementations for MiniScala and Lightweight-Modular-Staging (LMS). %
Cong showed that it is possible to combine the simplicity of let-bindings
for sequential execution and the power of continuations for further
control flow, by adding control operators to ANF,
which enable to capture the current continuation.

In our work,
we highlight the importance of commuting conversions,
and extend the idea of having an intermediate language that intrinsically
encodes maximal let-let conversion
to an intermediate language that also intrinsically encodes
maximal let-for and let-if commuting conversion, without exponential blow-up of code size.
Further, as this work lies in the context of array programming,
recursion is not often necessary, and can thus be avoided.

\subsection{Array Programming}

Shaikhha et al.~\cite{Shaikhha19} present a differentiable programming language which is an
extension of the lambda calculus. For array computations, they use an approach based on
higher-order functions. It directly represents the duality of functions and arrays
through built-in functions |build| for creating an array from a function and |get| for
turning an array into a function. The fact that |get| is a left inverse of |build| leads
to the equivalence |get (build n e) i ≡ e i|, which can be used for optimization.
Another strand of research makes use of the standard technique of \textit{rewriting strategies} to
optimize functional array programs~\cite{Hagedorn20,Bohler23}, suggesting the viability of
standard techniques from term rewriting for optimizing array programs.
Liu et al.~\cite{Liu22} present a framework that can express a variety of optimizations through formally verified term rewriting,
achieving competitive performance; however, CSE is not addressed.
Their representation is first-order and
features an array generation construct similar to the one in \Lang{}.
Optimization in \Lang{} is not based on rewriting,
but instead uses partial evaluation.

Feldspar~\cite{Axelsson10} is a DSL for array computations in Haskell.
It features a |parallel| construct similar to our |for| constructor,
as well as |while| loops. Feldspar is compiled to C and performs standard optimizations like fusion,
as well as copy propagation and loop unrolling. Feldspar's backend uses a dataflow graph and
an imperative intermediate representation, whereas our intermediate representation is functional and specifically
designed to support optimizations on array programs.

SaC~\cite{Scholz03} is a functional first-order array language. Array
computations are expressed using \textit{with-loops}, which consist of at least one generator
and one operator. Each generator consists of an index range and an expression giving the value
of the output array at a given index in the range. The operator can provide default values
for indices not included in any generator, a base array that should be modified by the
generators, or it can describe an aggregation.
More recently, SaC has added support for \textit{tensor comprehensions}~\cite{Scholz19},
which drop the operator part and add pattern matching on indices as well as bound and shape
inference, making the notation more lightweight. Similar to our approach, their
tensor comprehensions do not support summation, which is added in the form of a built-in
function.
SaC's optimizations are not based on a logical foundation,
but consist of a pipeline of optimization algorithms.
In \Lang{} we derive our syntax form for pattern matching on arrays from polarization type theory, enabling additional commuting conversions and thus grounding our optimization algorithm on a logical foundation.

\section{Conclusion}
\label{sec:conclusion}

This paper introduced \AINF{}, a novel intermediate representation for array computations, and \Lang{}, a surface array language.
The proposed optimization algorithm for \AINF{}, based on typed partial evaluation and common subexpression elimination,
simplifies program optimization by interpreting arrays as positively polarized types.
This approach avoids complexities associated with optimization schedules for conventional ANF.
We formalized \AINF{} and \Lang{}. We proved sound the translation from \Lang{} to \AINF{} and optimization.

For future work,
we are working on extending the language with automatic differentiation and probabilistic primitives,
and proving these extensions correct as well.
We are interested in applying our optimization to redundancies generated by automatic differentiation.

\bibliography{bib.bib}

\end{document}